\documentclass[a4paper]{article}
\usepackage{amsmath,amsthm,amssymb,bm, comment}
\usepackage{authblk}
\usepackage{xcolor}
\usepackage{graphicx} 
\usepackage{hyperref}
\usepackage{booktabs}
\usepackage{adjustbox}
\usepackage{multirow}
\usepackage[normalem]{ulem}

\usepackage[linesnumbered,ruled,vlined]{algorithm2e}
\usepackage{listings}

\newcommand{\cD}{\mathcal{D}}

\newcommand{\cF}{\mathcal{F}}

\newcommand{\bbC}{\mathbb{C}}

\newcommand{\bbQ}{\mathbb{Q}}

\def\dres{\partial{\rm Res}}

\def\0{{\bf 0}}

\def\dres{\partial{\rm Res}}
\def\dsres{\partial{\rm sres}}

\newtheorem{thm}{Theorem}[section]

\newtheorem{lem}[thm]{Lemma}

\theoremstyle{definition}
\newtheorem{defi}[thm]{Definition}

\newtheorem{example}[thm]{Example}

\theoremstyle{remark}

\SetKwInOut{Input}{Input}
\SetKwInOut{Output}{Output}
\SetKwComment{Comment}{// }{}

\date{} %

\title{Computing differential subresultants with Maple}

\author{
Marcos Cabellos\\
Universidad Polit\' ecnica de Madrid.\\
\tt{marcos.cabellos@alumnos.upm.es}
\and
Sonia L. Rueda\\
Dpto. de Matem\' atica Aplicada, E.T.S. Arquitectura\\
Universidad Polit\' ecnica de Madrid.\\
Avda. Juan de Herrera 4, 28040-Madrid, Spain.\\
\tt{sonialuisa.rueda@upm.es}
}

\begin{document}

\maketitle

\begin{abstract}
    We review the definition and main properties of differential subresultants in order to achieve their implementation in Maple, using the DEtools package. The focus is on computing GCRDs of ordinary differential operators with non necessarily rational coefficients. Determinant expressions provide explicit control, enabling the treatment of coefficients with parameters. Applications to commuting ordinary differential operators illustrate the effectiveness of the method.
\end{abstract}

\textsc{\href{https://zbmath.org/classification/}{MSC[2020]}: 13P15, 12H05.}



\section{Introduction}

In the symbolic analysis of differential equations, the computation of greatest common right divisors (GCRDs) of differential operators plays a central role in solving operator equations, simplifying systems, and analyzing structural properties of differential ideals. Classical subresultant theory provides an efficient and well-understood algebraic framework for computing polynomial greatest common divisors  \cite{Lalo91} . It has been  extended to the non-commutative setting in the case of ordinary differential operators \cite{Cha} and more generally Ore polynomials \cite{Li}.  

This article provides an overview of the construction of differential subresultants and their main properties, focusing on their implementation in  Maple. The main goal is to make these tools available to the symbolic analysis community, allowing the computation of GCRDs of ordinary differential operators whose coefficients are not necessarily rational functions. Having control over the computation process of  coefficients, as determinants of submatrices of the Sylvester style matrices, enhances the possibilities of including parameters and further extracting conclusions. This was the case in recent works framed in the theory of commuting differential operators, starting with \cite{MRZ1} and summarized in \cite{Ru25}.

We present an implementation of differential subresultants, using the Maple computer algebra system, in the form of a library DSres. Specifically, it leverages the capabilities of the DEtools package, which provides built-in support for differential operators and related symbolic manipulations. We will review definition and main properties of subresultants in Section \ref{Sec2}, followed by an exploration of the package DSres in Section \ref{SecMaple}. We finish with two examples of application in the theory of commuting differential operators in Section \ref{SecEx}.


\noindent {\bf Preliminaries.} For concepts in differential algebra, we refer to
\cite{Ritt, Kolchin, VPS}.
A {\it differential ring} is a ring $R$ with a derivation $\partial$ on $R$. A differential ideal $I$ is an ideal of $R$ invariant under the derivation. 
Assuming that $R$ is a differential domain, its field of fractions $Fr (R)$ is a differential field with extended derivation
$$\partial (f/g) = ( \partial (f)g-f\partial (g))/g^2.$$
A {\it differential field} $(K, \partial )$ is a differential ring which is a field. Given $a\in K$, we denote $\partial(a)$ by $a'$. 

\section{Resultants and subresultants}\label{Sec2}

The Sylvester style formulas for resultants and subresultants were introduced by M. Chardin in \cite{Cha}, and studied for Ore polynomials by Z. Li in \cite{Li}. In this section we review these results with the main goal of preparing for their Maple implementation.


Let us consider ordinary differential operators (ODOs) whose coefficients belong to a differential domain $\mathbb{D}$ with derivation $\partial$. When wondering about the common factors of two of these ODOs, it is convenient to consider the fraction field $\mathbb{K}$ of $\mathbb{D}$, with the natural extension of the derivation, still denoted as $\partial$. This will allow us to ensure that the ring of ODOs $\mathbb{K}[\partial]$ is a left and right Euclidean domain, and as such, a principal ideal domain and a unique factorization domain \cite{BGV}. Given $A,B \in \mathbb{K}[\partial]$, we will denote their {\it greatest common right divisor} as $\textrm{gcrd}(A,B)$. Furthermore, $A$ and $B$ will be {\it right comprime} when $\textrm{gcrd}(A,B)\in \mathbb{K}$.  


The differential resultant appears in this context with similar goals as the algebraic resultant, when instead of ODOs we contemplate polynomials in one variable. There exists here, as in the polynomial case \cite{Cox}, an equivalence between the presence of common non-trivial factors of $A,B$ and the existence of order-bounded $P,Q \in \mathbb{K}[\partial]$ such that $PA+QB=0$, see \cite{McW}.  
To construct the resultant, let us assume from now on that $n:=\textrm{order}(A)$ and $m:=\textrm{order}(B)$. This inspires the linear map
\begin{equation}
    s_0: \mathcal{V}_{n} \oplus \mathcal{V}_m\rightarrow \mathcal{V}_{n+m}
\end{equation}
\begin{displaymath}
     (P,Q)\mapsto P A+Q B
\end{displaymath}

where $\mathcal{V}_k$ stands for the $\mathbb{K}$-vector subspace of $\mathbb{K}[\partial]$ whose elements have order strictly bounded by $k$. The matrix of this linear map, in the canonical basis $\mathcal{B}_k = \{1,\partial,...\partial ^{k-1}\}$ for $\mathcal{V}_k$, is an $(n+m)\times (n+m)$ matrix that will be refered to as the {\it (differential) Sylvester matrix} of $A$ and $B$ or $\textrm{Syl}(A,B)$, and may be expressed row by row by the coefficients of the following ODOs in the basis $\mathcal{B}_k$
\begin{equation}\label{def_sylv}
    \textrm{Syl}(A,B) =\left( \begin{array}{c|c|c|c|c|c|c|c}
         [\partial^{m-1} A]_{\mathcal{B}_{m+n}}& 
         [\partial^{m-2} A]_{\mathcal{B}_{m+n}}&
         \dots &
         [A]_{\mathcal{B}_{m+n}}&
         [\partial^{n-1} B]_{\mathcal{B}_{m+n}}&
         ... &
         [B]_{\mathcal{B}_{m+n}}
    \end{array} \right)^T.
\end{equation}

\begin{defi}
    Given $A,B\in \mathbb{D}[\partial]$, the {\sf differential resultant} of $A,B$, denoted $\dres(A,B)$, is defined to be \begin{displaymath}
            \dres(A,B) := \textrm{det}(\textrm{Syl}(A,B)).
    \end{displaymath}
\end{defi}

To properly utilize this new object, it is important to understand where does it map the pair $(A,B)$. The proof of the next theorem can also be found in \cite{Cha, Li, McW}.

\begin{lem}\label{lema}
    Given $A,B\in \mathbb{D}[\partial]$, their differential resultant stays in the image of $s_0$. That is,

    \begin{displaymath}
        \dres(A,B) = PA+QB \in \textrm{Im}(s_0) \cap \mathbb{D} 
    \end{displaymath}
    for some pair $P,Q\in \mathbb{D}[\partial]$ such that order$(P) = m-1$ and order$(Q)=n-1$.
\end{lem}
\begin{proof}
    We can make some elemental transformations to the determinant that defines the resultant to obtain said $P,Q$. We will add to the last column of $\textrm{Syl}(A,B)$ every other column $C_i$ for $i=1,...,n+m-1$ multiplied by a factor $\partial ^{n+m-i}$. The resulting matrix by columns would be
    \begin{equation}\label{AddCol}
        \left(\begin{array}{c|c|c|c|c}
            C_1 & C_2 & ... & C_{n+m-1} & \sum_{i=1}^{n+m} C_i \partial^{n+m-i}  
        \end{array}\right)
    \end{equation}
    and this new last column is, because of the structure noted in \eqref{def_sylv},
    \begin{equation}
        \sum_{i=1}^{n+m} C_i \partial^{n+m-i} = \left(\begin{array}{c|c|c|c|c|c|c|c}
         \partial^{m-1} A  &
         \partial^{m-2} A&
         ...&
         A&
         \partial^{n-1}B&
         \partial^{n-2}B&
         ... &
         B
        \end{array}\right) ^T.
    \end{equation}
    Therefore, developing the new determinant by the last column and extracting the right factors $A$ and $B$ we will obtain
    \begin{equation}
        \dres(A,B) = (\sum_{i=1}^m M_i\partial^{m-i} )A + (\sum_{i=1}^n M_{i+m}\partial^{n-i})B 
    \end{equation}
    where $M_k$ is the $k$-th minor in the development, and thus $P:=\sum_{i=1}^m M_i\partial^{m-i}, Q:=\sum_{j=1}^n M_{m+j}\partial^{n-j}$ are the ODOs of the order we needed. Note that $M_k\in \mathbb{D}$ since $\textrm{Syl}(A,B)$ was constructed through the coefficients of ODOs in $\mathbb{D}[\partial]$. 
\end{proof}

Through this lemma we can finally link the resultant to the non-triviality of $\textrm{gcrd}(A,B)$.

\begin{thm}\label{teorema_resultante}
    Given $A,B\in \mathbb{D}[\partial]$, 
    \begin{equation}
        \dres(A,B) \neq 0 \Longleftrightarrow A,B\textrm{ are right coprime in } \mathbb{K}[\partial].
    \end{equation}
\end{thm}
\begin{proof}

    Given $A,B$ such that their resultant is non-zero, it is evident from Lemma \ref{lema} that  $PA+QB=d\in\mathbb{D}$ for some adequate $P,Q$. Considering $F:=\textrm{gcrd}(A,B)$ and $A=\Tilde{A}F,B=\Tilde{B}F$ then
    \begin{displaymath}
        d=PA+QB=(P\Tilde{A}+Q\Tilde{B})F.
    \end{displaymath} Because $\mathbb{D}$ is a domain and order$(d)=0$, $F$ must be of order zero as well.
    
    The remaining implication holds since $\mathbb{K}[\partial]$ is right Euclidean. It suffices to apply Euclid's extended algorithm (or directly summoning Bézout's identity) to get $PA+QB=\textrm{gcrd}(A,B)\in \mathbb{K}$ for some $(P,Q)\in \mathcal{V}_{n} \oplus \mathcal{V}_m$, and thus $s_0(P\textrm{gcrd}(A,B)^{-1},Q\textrm{gcrd}(A,B)^{-1})=1$, making $s_0$ surjective and its determinant, the resultant, non-zero.
\end{proof}
    
    

Summing up, the differential resultant will indicate if two ODOs are right coprime or if they have some non-trivial common factor. To obtain such factor explicitly we may contemplate  subresultants as an useful tool.

The differential resultant can be thought of as the $0$-th differential subresultant. Let us construct explicitly the first subresultant. To do so, we will define $s_1$ as

\begin{equation}
    s_1: \mathcal{V}_{n-1} \oplus \mathcal{V}_{m-1}\rightarrow \mathcal{V}_{n+m-1}
\end{equation}
\begin{displaymath}
     (P,Q)\mapsto P A+Q B.
\end{displaymath}

Put into perspective, we are bounding the orders of the ODOs of the linear combination. This new linear map presents a similarly structured matrix $\mathcal{M}_1$ in canonical basis, but now rectangular of size $(n+m-2)\times(n+m-1)$
\begin{equation*}
    \mathcal{M}_1 =\left( \begin{array}{c|c|c|c|c|c|c|c}
         [\partial^{m-2} A]_{\mathcal{B}_{m+n-1}}& 
         [\partial^{m-3} A]_{\mathcal{B}_{m+n-1}}&
         \dots &
         [A]_{\mathcal{B}_{m+n-1}}&
         [\partial^{m-2} B]_{\mathcal{B}_{m+n-1}}&
         ... &
         [B]_{\mathcal{B}_{m+n-1}}
    \end{array} \right)^T
\end{equation*}
which is equivalent to eliminating the topmost and $(m+1)$-th rows together with the leftmost column from $\textrm{Syl}(A,B)$. To obtain the first subresultant we must compute $\mathcal{M}_1$'s determinant polynomial, as defined in \cite{Li}. Given a $r\times c$ matrix $\mathcal{M}$ with $r\leq c$, its polynomial determinant pdet($\mathcal{M}$) is obtained as
\begin{equation}
    \textrm{pdet}(\mathcal{M}) = \sum_{i=0}^{c-r}\textrm{det}(\mu_i)\partial^i
\end{equation}
where $\mu_i =(C_1|...|C_{r-1}|C_{c-i})$, when $C_j$ is the $j$-th column of $\mathcal{M}$. Notice that, inserting $\partial^i$ in the rightmost column, we can group it into a single determinant. We can also add each column to the last one as we did in \eqref{AddCol} resulting in
\begin{equation}\label{onedet}
    \textrm{pdet}(\mathcal{M}) = \textrm{det}(C_1|...|C_{r-1}|\sum_{i=0}^{c-r}C_{c-i}\partial ^i) = \textrm{det}(C_1|...|C_{r-1}|\sum_{i=0}^{c-1}C_{c-i}\partial^i).
\end{equation}

This means that the first differential subresultant of $A$ and $B$, $\partial\textrm{sres}_1(A,B)$ is computed as
\begin{equation}
    \partial\textrm{sres}_1(A,B) = \textrm{pdet}(\mathcal{M}_1).
\end{equation}

Now we can easily extend this definition to the $i$-th subresultant with $i=0,...,m-1$. We construct the linear map $s_i$ as
\begin{equation}
    s_i: \mathcal{V}_{n-i} \oplus \mathcal{V}_{m-i}\rightarrow \mathcal{V}_{n+m-i}
\end{equation}
\begin{displaymath}
     (P,Q)\mapsto P A+Q B.
\end{displaymath}
whose matrix in the canonical basis can be explicitly described as
\begin{equation*}\label{explicit}
    \mathcal{M}_i =\left( \begin{array}{c|c|c|c|c|c|c|c}
         [\partial^{m-i-1} A]_{\mathcal{B}_{m+n-i}}& 
         [\partial^{m-i-2} A]_{\mathcal{B}_{m+n-i}}&
         \dots &
         [A]_{\mathcal{B}_{m+n-i}}&
         [\partial^{n-i-1} B]_{\mathcal{B}_{m+n-i}}&
         ... &
         [B]_{\mathcal{B}_{m+n-i}}
    \end{array} \right)^T
\end{equation*}
or, starting from $\textrm{Syl}(A,B)$, erase the first $i$ rows, the $i$ rows following the $(m+1)$-th, and the leftmost $i$ columns.

\begin{defi}
    Given $A,B\in \mathbb{D}[\partial]$, their $i$-th {\sf differential subresultant}, denoted $\dsres_i(A,B)$, is defined to be \begin{displaymath}
        \dsres_i(A,B):=\textrm{pdet}(\mathcal{M}_i).
    \end{displaymath}    
\end{defi}

Now we should prove that these tools do indeed offer a way to compute $\textrm{gcrd}(A,B)$, by the so called subresutant algorithm, proved in \cite{Cha, Li}, and included here for completion.

\begin{thm}\label{thm-subres}
    Given $A,B\in \mathbb{D}[\partial]$. The following are equivalent:
    \begin{enumerate}
        \item $d:=\textrm{order}(\textrm{gcrd}(A,B))$.

        \item  $\dsres_i(A,B)=0$ for $0\leq i < d$  and $\dsres_d(A,B)\neq 0$.
    \end{enumerate}
    In this situation, the monic $\textrm{gcrd}(A,B)$ of $A$ and $B$ in $\mathit{K}[\partial]$ equals $\dsres_d(A,B)$ divided by its leading coefficient.
\end{thm}
\begin{proof}
    It is enough to prove that (1) implies (2). We may proceed in a similar way as we did for the differential resultant. First, we must show for any subresultant that 
    \begin{displaymath}
        \dsres_i(A,B)\in\textrm{Im}(s_i) \cap \mathcal{V}_{i+1}. 
    \end{displaymath} 
    Note that the difference with the resultant is that now the result is an ODO of bounded order. However the argument is almost the same as before, since we can rearrange the polynomial determinant. Looking at $\mathcal{M}_i$, we can see that when applying pdet as in \eqref{onedet}, the last column is simply
    \begin{equation}
        \sum_{j=0}^{n+m-i}C_{n+m-i-j}\partial^j = (\partial^{m-1-i}A |...| A |\partial ^{n-1-i}B|...|B)^T.
    \end{equation}
    Developing the determinant we reach
    \begin{equation}
        \dsres_i(A,B)=(\sum_{j=1}^{m-i}M_j\partial^j)A + (\sum_{j=1}^{n-i}M_{j+m-i}\partial^j)B \in \textrm{Im}(s_i) \cap \mathcal{V}_{i+1} .
    \end{equation}
    With this property we can now see that if $i < d$, then for some $(P,Q)\in \mathcal{V}_{n-i} \oplus \mathcal{V}_{m-i}$
    \begin{equation}\label{eq-dsresi}
        \dsres_i(A,B) = PA+QB.
    \end{equation}
    Since $\textrm{gcrd}(A,B)$ can divide the right hand side of the equality as a right Euclidean domain $\mathbb{K}[\partial]$, it should do so too for the subresultant. However, since order$(\partial\textrm{sres}_i(A,B))\leq i < d$, the subresultant must be zero. If however $i=d$, as per this last argument we only need to prove it is non-zero. Let us prove it by reduction to the absurd. If indeed we had a null subresultant, it would mean that all the determinants that compose it are zeroed. But if that were the case, we could construct $(P,Q)\in \mathcal{V}_{n-i} \oplus \mathcal{V}_{m-i}$ superposing all $i$ linear relations implied by the determinants such that $s_i(P,Q)=PA+QB=0$. However, we know that 
    \begin{equation}\label{eq-PAQB}
        PA+QB = (P\tilde{A}+Q\tilde{B})\textrm{gcrd}(A,B).
    \end{equation}
    If we link both equalities, 
    \begin{displaymath}
    (P\Tilde{A}+Q\Tilde{B})\textrm{gcrd}(A,B) = 0,
    \end{displaymath}
    so because we operate in a domain, either the common divisor is zero, which is not the case; or the other factor is. However, applying Theorem \ref{teorema_resultante} to it since we now have the adequate orders to do so, this is a contradiction as we have a null resultant of right coprime operators $\Tilde{A},\Tilde{B}$.

    Observe that if the statements hold then by \eqref{eq-dsresi} and \eqref{eq-PAQB} then the monic \textrm{gcrd}(A,B) is the $\dsres_d(A,B)$ divided by its leading coefficient.
\end{proof}

\section{Maple implementation}\label{SecMaple}

We present the Maple package DSres that contains the tools to compute differential resultant and differential subresultants. This section contains the description of the functions of the package. They are based on the Maple packages DEtools, for manipulation of differential operators, essentially for multiplication, and LinearAlgebra for computation of determinants. All functions have the additional argument $t$ to indicate that work is carried out in a ring of differential operators $\mathbb{D}[\partial]$ with derivation $\partial=d/dt$. 

Another Maple package that allows multiplications in $\mathbb{D[\partial]}$ is OreTools, but DEtools was chosen to define ODOs as polynomials in $Dt$. Both DEtools and OreTools allow for the computation of greatest common right divisors, however in many occasion more control over the computation process of gcrd is needed, as in the case of ODOs including parameters or coefficients that are non necessarily rational functions.



The package DSres and all the examples presented in this paper are available at

\begin{center}
    \url{https://github.com/MarcosCH04/DSres/tree/main}
\end{center}



The commands of the package DSres are:

\begin{itemize}
    \item \textbf{DSylvester} is called as
    \begin{center}
        DSylvester$(A,B,t)$
    \end{center}
    and will output $\textrm{Syl}(A,B)$. 

    \item \textbf{DResultant} is called as
    \begin{center}
        DResultant$(A,B,t)$
    \end{center}
    and will output $\dres(A,B)$.
    
    \item \textbf{DSubresultantMatrix} is called as
    \begin{center}
        DSubresultantMatrix$(A,B,ind,degA,degB,t,dim)$
    \end{center}
    and will output the matrix $\mathcal{M}_{ind}$ as described in the previous section. This process has mainly the goal to allow the subresultant-related processes to work as intended, but we have opted to allow user interaction since the matrix it offers can be of interest. The remaining arguments are
    \begin{itemize}
        \item \textbf{degA} the order of $A$ in $DA$. It \textbf{must be provided}.
        \item \textbf{degB} the order of $B$ in $DA$. It \textbf{must be provided}.
        \item \textbf{dim} will always be $degA+degB-2ind$. It \textbf{must be omitted} and is only there to avoid repeating calculations in following processes. 
    \end{itemize}
    These arguments are not computed internally because this process is used in the following ones and would calculate these same values many times.

    \item \textbf{DSubresultant} is called as
    \begin{center}
        DSubresultant$(A,B,ind,t)$
    \end{center}
    and will output a two-elements list, in order, $\dsres_{ind}(A,B)$ and a nested list with all matrices involved in the computation of the polynomial determinant that defines the subresultant.

    \item \textbf{DSubresultantSequence} is called as
    \begin{center}
        DSubresultantSequence$(A,B,t)$
    \end{center}
    and will output a list of all subresultants and the matrices used in their computation. Precisely, it has $\textrm{min}\{\textrm{order}(A),\textrm{order}(B)\}$ elements, each of them a list obtained from applying the previous procedure DSubresultant.
\end{itemize}

\begin{example}
Let us consider differential operators  $A=\sum_{i=1}^n a_i \partial^i$ and $B=\sum_{j=1}^m b_j \partial^j$ whose coefficients are differential variables. We can consider the differential domain
\begin{equation}
    \cD=\bbQ\{a_i,b_j\}=\bbQ[a_i, a_i^{[k]}:=\partial^k(a_i), b_j, b_j^{[k]}:=\partial^{k} (b_j)].
\end{equation}

One important property of differential resultant and subresultant is that their coefficients stay in $\cD$. We illustrate this fact for $n=2$ and $m=3$.

\begin{lstlisting}
> with(DSres):with(DEtools):
> # We define generic operators A and B in the differential algebra DA.
> DA:=[Dx,x]:
> A:=sum(a[i](x)*Dx^i,i=0..2):B:=sum(b[j](x)*Dx^j,j=0..3):
> # Computing differential Sylvester matrix.  
> DSylvester(A,B);
\end{lstlisting}
\textcolor{blue}{
$$\label{(1)}
    \left[\begin{array}{ccccc}
a_{2}\! \left(x \right) & 2 \frac{d}{d x}a_{2}\! \left(x \right)+a_{1}\! \left(x \right) & \frac{d^{2}}{d x^{2}}a_{2}\! \left(x \right)+2 \frac{d}{d x}a_{1}\! \left(x \right)+a_{0}\! \left(x \right) & \frac{d^{2}}{d x^{2}}a_{1}\! \left(x \right)+2 \frac{d}{d x}a_{0}\! \left(x \right) & \frac{d^{2}}{d x^{2}}a_{0}\! \left(x \right) 
\\
 0 & a_{2}\! \left(x \right) & \frac{d}{d x}a_{2}\! \left(x \right)+a_{1}\! \left(x \right) & \frac{d}{d x}a_{1}\! \left(x \right)+a_{0}\! \left(x \right) & \frac{d}{d x}a_{0}\! \left(x \right) 
\\
 0 & 0 & a_{2}\! \left(x \right) & a_{1}\! \left(x \right) & a_{0}\! \left(x \right) 
\\
 b_{3}\! \left(x \right) & \frac{d}{d x}b_{3}\! \left(x \right)+b_{2}\! \left(x \right) & \frac{d}{d x}b_{2}\! \left(x \right)+b_{1}\! \left(x \right) & \frac{d}{d x}b_{1}\! \left(x \right)+b_{0}\! \left(x \right) & \frac{d}{d x}b_{0}\! \left(x \right) 
\\
 0 & b_{3}\! \left(x \right) & b_{2}\! \left(x \right) & b_{1}\! \left(x \right) & b_{0}\! \left(x \right) 
\end{array}\right]
$$
}

\begin{lstlisting}
> # Computing the first subresultant.
> DSres1:=DSubresultant(A,B,1):
> # The first element of the output list stores the subresultant.
> DSres1[1]:
\end{lstlisting}

In the notation of the differential algebra $\cD$
\begin{align*}\dsres_1(A,B)=&\det(S_0^1)+ \det(S_1^1) \partial,\\
&\det(S_0^1)=a_{2}^{2} b_{0}-a_{2} a_{0} b_{2}-a_{2} b_{3} a_{0}'+a_{0} b_{3} a_{2}'+a_{0} b_{3} a_{1},\\
&\det(S_1^1)=b_{1} a_{2}^{2}-a_{1}' a_{2} b_{3}-a_{2} a_{0} b_{3}-a_{2} b_{2} a_{1}+b_{3} a_{2}' a_{1}+b_{3} a_{1}^{2}. 
\end{align*}
The list of matrices $[S_0^1,S_1^1]$ is stored in the second element of the output list. 

\begin{lstlisting}
> DSres1[2];
\end{lstlisting}
\textcolor{blue}{
$$\label{(2)}
    \left[
\left[\begin{array}{ccc}
a_{2}\! \left(x \right) & \frac{d}{d x}a_{2}\! \left(x \right)+a_{1}\! \left(x \right) & \frac{d}{d x}a_{0}\! \left(x \right) 
\\
 0 & a_{2}\! \left(x \right) & a_{0}\! \left(x \right) 
\\
 b_{3}\! \left(x \right) & b_{2}\! \left(x \right) & b_{0}\! \left(x \right) 
\end{array}\right]
, 
\left[\begin{array}{ccc}
a_{2}\! \left(x \right) & \frac{d}{d x}a_{2}\! \left(x \right)+a_{1}\! \left(x \right) & \frac{d}{d x}a_{1}\! \left(x \right)+a_{0}\! \left(x \right) 
\\
 0 & a_{2}\! \left(x \right) & a_{1}\! \left(x \right) 
\\
 b_{3}\! \left(x \right) & b_{2}\! \left(x \right) & b_{1}\! \left(x \right) 
\end{array}\right]
\right]
$$
}
The matrices $S_0^1$ and $S_1^1$ are submatrices of $\mathcal{M}_1$ respectively removing the columns indexed by $\partial$ and $1$. The matrix $\mathcal{M}_1$ can be accessed as 
\begin{lstlisting}
> DSres1_total:=DSubresultantMatrix(A,B,1,2,3):
\end{lstlisting}
\textcolor{blue}{
$$\label{(3)}
    \left[\begin{array}{cccc}
a_{2}\! \left(x \right) & \frac{d}{d x}a_{2}\! \left(x \right)+a_{1}\! \left(x \right) & \frac{d}{d x}a_{1}\! \left(x \right)+a_{0}\! \left(x \right) & \frac{d}{d x}a_{0}\! \left(x \right) 
\\
 0 & a_{2}\! \left(x \right) & a_{1}\! \left(x \right) & a_{0}\! \left(x \right) 
\\
 b_{3}\! \left(x \right) & b_{2}\! \left(x \right) & b_{1}\! \left(x \right) & b_{0}\! \left(x \right) 
\end{array}\right]
$$
}

\end{example}

\begin{example}
All the functions of the package DSres allow the declaration of the variable name "t" to define the generic ring of differential operators $[Dt,t]$. If omitted it is taken as $[Dx,x]$. In the next example it is declared as $[Ds,s]$. 

\begin{lstlisting}
> with(DSres):with(DEtools):
> A := Ds^2+s : B:= Ds^3+Ds:
\end{lstlisting}
\begin{lstlisting}
> # The parameter specifies a differential algebra [Ds,s].
DResultant(A,B,s);
\end{lstlisting}
\textcolor{blue}{
$$\label{(4)}
s^{3}-2 s^{2}+s +2
$$
}
\end{example}

\section{Applications in the theory of commuting differential operators}\label{SecEx}

We will demonstrate the usefulness of having explicit formulas for differential resultants and subresultants with examples from the theory of commuting differential operators, including important developments from the beginning of the twentieth century by Schur, Burchnall-Chaundy, Krichever, Mumford, Wilson, Previato to name a few from a long list \cite{Zheglov}, \cite{Ru25}.


Consider differential operators $A$ and $B$ whose coefficients belong to a differential field $K$ whose field of constants is the field of complex numbers $\bbC$. Let us assume that $A$ and $B$ commute. By the theory of commuting differential operators, the $\bbC$-algebra of commuting differential operators $\bbC[A,B]$ is isomorphic to the coordinate ring of an algebraic curve $\Gamma$ 
\begin{equation}
     \bbC[\Gamma]:=\frac{\bbC[\lambda,\mu]}{(f(\lambda,\mu))}\simeq \bbC[A,B],
\end{equation}
for algebraic parameters $\lambda$ and $\mu$, that is $\partial(\lambda)=\partial(\mu)=0$.
 The famous spectral curve $\Gamma$ is defined by the polynomial $f\in \bbC[\lambda,\mu]$ 
\[\Gamma=\{(\lambda_0,\mu_0)\in \bbC^2\mid f(\lambda_0,\mu_0)=0\}.\]
It is well known that $f$ is the Burchnall-Chaundy polynomial of $A$ and $B$, that is $f(A,B)=0$, and that $f$ equals $\sqrt{h}$, the square free part of the differential resultant
\begin{equation}
    h(\lambda,\mu)=\dres(A-\lambda,B-\mu)\in \bbC[\lambda,\mu],
\end{equation}
which is a non zero polynomial with constant coefficients.

By Theorem \ref{teorema_resultante}, as differential operators in $K(\lambda,\mu)$ then $A-\lambda$ and $B-\mu$ are right coprime. Consider now a new coefficient ring 
\begin{equation}
    K[\Gamma]=\frac{K[\lambda,\mu]}{[f]},
\end{equation}
where $[f]$ is now the ideal generated by $f$ in $K[\lambda,\mu]$, which a differential ring for an extended derivation $\tilde{\partial}$ defined by $\tilde{\partial} (g+(f))=\partial(g)+[f]$, that we denote again by $\partial$.

The greatest common right divisor of $A-\lambda$ and $B-\mu$ as differential operators with coefficients in $K(\Gamma)=\mathrm{Fr}(K[\Gamma])$ is the vector bundle $\cF(\lambda,\mu)$ of the spectral curve $\Gamma$ \cite{PW}, that we will compute next using differential subresultants. 

\subsection{Euler Operators}

Let us consider a family of differential operators in $\bbC(x)[\partial]$, with $\partial=d/dx$, called Euler operators  
$$E_n := x^{-n}\delta (\delta-m)(\delta -2m)\ldots (\delta -m(n -1)), \mbox{ for } \delta := x \partial.$$
It is proved in \cite{MP2023} that $E_n$ commutes with $E_m$ and $E_n^m=E_m^n$.  

For $A=E_4$ and $B=E_6$, let us compute their differential resultant.

\begin{lstlisting}
> with(DSres): with(DEtools):
> # We define Euler Operators in the following differential algebra.
> DA:=[Dx,x]:
\end{lstlisting}

\begin{lstlisting}
> EulerOp:=proc(N,M)
> local L,i; 
> L:=x^(-N)*x*Dx; for i from 1 to N-1 do L:=mult(L,x*Dx-i*M,DA);od;
> return L;
> end:
> 
> A:=EulerOp(4,6);B:=EulerOp(6,4);
\end{lstlisting}
\textcolor{blue}{
$$\label{(5)}
A := \mathit{Dx}^{4}-\frac{30 \mathit{Dx}^{3}}{x}+\frac{295 \mathit{Dx}^{2}}{x^{2}}-\frac{935 \mathit{Dx}}{x^{3}}$$
}
\textcolor{blue}{
$$\label{(6)}
B := \mathit{Dx}^{6}-\frac{45 \mathit{Dx}^{5}}{x}+\frac{825 \mathit{Dx}^{4}}{x^{2}}-\frac{7650 \mathit{Dx}^{3}}{x^{3}}+\frac{35595 \mathit{Dx}^{2}}{x^{4}}-\frac{65835 \mathit{Dx}}{x^{5}}
$$
}

\begin{lstlisting}
> # Computing differential resultant.
dresult:=DResultant(A-lambda,B-mu): 
simplify(dresult);
\end{lstlisting}
\textcolor{blue}{
$$\label{(7)}
\left(-\lambda^{3}+\mu^{2}\right)^{2}
$$
}

The spectral curve $\Gamma$ of the pair $A, B$ is thus defined by $f(\lambda,\mu)=\lambda^3 - \mu^2$. Consider the new coefficient field $K(\Gamma)$ with $K=\bbC(x)$, which is the fraction field of the domain
\[K[\Gamma]=\frac{K[\lambda,\mu]}{[\lambda^3 - \mu^2]}.\]

Let us compute next the greatest common right divisor of $A-\lambda$ and $B-\mu$.

\begin{lstlisting}
> # Computing differential subresultant sequence.
dsres:=DSubresultantSequence(A-lambda,B-mu): 
# Visualize first subresultant
> simplify(dsres[2]);
\end{lstlisting}
\textcolor{blue}{
$$\label{(8)}
\frac{\left(\lambda^{3}-\mu^{2}\right) \left(\lambda  x^{4}-560\right)}{x^{4}},
$$
}
\begin{lstlisting}
# Visualize leading coefficient of second subresultant
> simplify(coeff(dsres[3],Dx^2));
\end{lstlisting}
\textcolor{blue}{
$$\label{(9)}
\frac{\left(\lambda x^4 -560\right)^2}{x^{8}}
$$
}

By the Theorem \ref{thm-subres}, the second differential subresultant is the greatest common right divisor of $A-\lambda$ and $B-\mu$ over $K(\Gamma)$, since it is their first non-zero subresultant 
\[\dsres_2(A-\lambda, B-\mu)=(\lambda x^4 - 560)\left( \frac{-\mu x^2 + 20\lambda}{x^6}+\frac{-5(3\lambda x^4 - 1232)}{x^9}\partial+\frac{\lambda x^4 - 560}{x^8} \partial^2\right),\]
since the leading coefficient is nonzero in $K(\Gamma)$.


It is a future project to implement a function that computes the greatest common right divisor of two operators in a differential field with algebraic parameters, whose relation is defined by an ideal generated by a list of polynomials in the given parameters.

One could compute the greatest common right divisor with DEtools but only if a parametrization of the spectral curve is provided. In this case $(s^2,s^3)$, $\forall s\in \bbC$

\begin{lstlisting}
> GCRD(A-s^2,B-s^3,DA);
\end{lstlisting}
\textcolor{blue}{
$$\label{(10)}
\mathit{Dx}^{2}-\frac{5 \left(3 s^{2} x^{4}-1232\right) \mathit{Dx}}{x \left(s^{2} x^{4}-560\right)}-\frac{s^{2} \left(s \,x^{2}-20\right) x^{2}}{s^{2} x^{4}-560}
$$
} 
\begin{lstlisting}
> GCRD(A-lambda,B-mu,DA);
\end{lstlisting}
\textcolor{blue}{
$$\label{(11)}
1
$$
} 

\subsection{Lamé operators}

By means of the DSres package, one can compute the greatest common right divisor of differential operators whose coefficients are not rational functions. 

Let us consider the Lamé operator $A=\partial^2-2\wp (x)$, where $\wp=\wp(x;g_2,g_3)$ is the Weierstrass $\wp$-function satisfying $(\wp')^2=4 \wp^3-g_2\wp-g_3$. The differential field of coefficients is $K=\bbC\langle \wp\rangle=\bbC(\wp,\wp')$ with derivation $\partial=d/dx$.

\begin{lstlisting}
> with(DSres):with(DEtools):
> # We construct differential operators A,B involving Weirstrass P function.
> DA:=[Dx,x];
\end{lstlisting}
\begin{lstlisting}
> wp:=WeierstrassP(x, g2,g3);wpp:=diff(wp,x):
\end{lstlisting}
\textcolor{blue}{
$$\label{(12)}
\mathit{wp} := \mathrm{WeierstrassP}\! \left(x ,\mathit{g2} ,\mathit{g3} \right)
$$
}
\begin{lstlisting}
> A:=Dx^2-2*wp: B:=-Dx^3+(3/2)*2*wp*Dx+(3/4)*2*wpp:
# Computing differential resultant.
dresultant:=DResultant(A-lambda,B-mu): 
simplify(dresultant);
\end{lstlisting}
\textcolor{blue}{
$$\label{(13)}
\frac{1}{4} \mathit{g3} +\mu^{2}+\frac{1}{4} \mathit{g2} \lambda -\lambda^{3}
$$
}

In this case the defining polynomial of the spectral curve is $f(\lambda,\mu)=\frac{1}{4} \mathit{g3} +\mu^{2}+\frac{1}{4} \mathit{g2} \lambda -\lambda^{3}
$. The greatest common right divisor is then the first subresultant.

\begin{lstlisting}
# Computing the first differential subresultant.
dsres1:=DSubresultant(A-lambda,B-mu,1)[1];
\end{lstlisting}
\textcolor{blue}{
$$\label{(14)}
-\frac{\mathrm{WeierstrassPPrime}\! \left(x ,\mathit{g2} ,\mathit{g3} \right)}{2}-\mu +\left(\mathrm{WeierstrassP}\! \left(x ,\mathit{g2} ,\mathit{g3} \right)  -\lambda\right)\mathit{Dx}  
$$
}

Even if we give a global parametrization of the curve $(\wp(x),\wp'(x)/2)$, DEtools cannot perform the greatest common right divisor, since it is designed for rational function coefficients.

\begin{lstlisting}
> GCRD(A-wp,B-(1/2)*wpp,DA);
\end{lstlisting}
\textcolor{blue}{
$$\label{(15)}
1
$$
}
\begin{lstlisting}
> DFactor(A-wp,DA);
\end{lstlisting}
$\href{http://www.maplesoft.com/support/help/errors/view.aspx?path=Error,%20(in%20DEtools%2FDFactor)%20expecting%20a%20rational%20function%20in%20%5BDx,%20x%5D%20but%20got%20Dx%5E2-3*WeierstrassP(x,g2,g3)}{Error, (in DEtools/DFactor) expecting a rational function in [Dx, x] but got Dx^2-3*WeierstrassP(x,g2,g3)}$

\medskip

Observe that DEtools cannot handle this factorization, whereas a right factor of $A-\wp$ was computed by means of  differential subresultants.

\medskip

As a future project, we will develop a new Maple command, based on DSres, to compute the greatest common right divisor of two differential operators $A-\lambda$ and $B-\mu$ over $K(\Gamma)$, including as an argument the ideal of the curve $\Gamma$. 

\medskip

\noindent{\bf Acknowledgments.} Supported by the grant PID2021-124473NB-I00, ``Algorithmic Differential Algebra and Integrability" (ADAI)  from the Spanish MCIN/AEI /10.13039/501100011033 and by FEDER, UE. In this special occasion, S.L. Rueda would like to thank L. Gonzalez-Vega  for introducing her to differential subresultants twenty years ago.

\bibliographystyle{plain}

\bibliography{Bibliography.bib}

\end{document}